\documentclass{article}
\usepackage{algorithm}
\usepackage{algpseudocode}
\usepackage[utf8]{inputenc}
\usepackage{amsmath}
\usepackage{amssymb}
\usepackage{amscd}
\usepackage{indentfirst}
\usepackage{graphicx}
\usepackage{graphics}
\usepackage{pict2e}
\usepackage{epic}
\usepackage{url}
\usepackage{epstopdf}
\usepackage{comment}
\usepackage{tikz}
\usepackage{xcolor}
\usepackage{fancyhdr}
\usepackage{color}
\usepackage{amsthm}
\usepackage{graphics}
\usepackage{graphicx}
\newtheorem{Th}{Theorem}[section]

\newtheorem{?}[Th]{Problem}

\title{Strong core and Pareto-optimal solutions for the multiple partners matching problem under lexicographic preferences}
\author{P\'eter Bir\'o and Gergely Cs\'aji}
\date{10 February 2022}

\begin{document}

\maketitle

\begin{abstract}
In a multiple partners matching problem the agents can have multiple partners up to their capacities. In this paper we consider both the two-sided many-to-many stable matching problem and the one-sided stable fixtures problem under lexicographic preferences. We study strong core and Pareto-optimal solutions for this setting from a computational point of view. First we provide an example to show that the strong core can be empty even under these severe restrictions for many-to-many problems, and that deciding the non-emptiness of the strong core is NP-hard. We also show that for a given matching checking Pareto-optimality and the strong core properties are co-NP-complete problems for the many-to-many problem, and deciding the existence of a complete Pareto-optimal matching is also NP-hard for the fixtures problem. On the positive side, we give efficient algorithms for finding a near feasible strong core solution, where the capacities are only violated by at most one unit for each agent, and also for finding a half-matching in the strong core of fractional matchings. These polynomial time algorithms are based on the Top Trading Cycle algorithm. Finally, we also show that finding a maximum size matching that is Pareto-optimal can be done efficiently for many-to-many problems, which is in contrast with the hardness result for the fixtures problem. 

\end{abstract}

\section{Introduction}

Roth \cite{Roth1984} proposed the study of many-to-many matching markets in the context of job markets, where each worker can have multiple jobs, and each firm can employ multiple workers, but at most one contract can be signed in between any worker and firm. The agents of such a market have choice functions over the possible contracts involving them, that specifies a subset for any given set of contracts. The most well studied solution concept is \emph{stability}. A solution is setwise stable if there are no alternative contracts outside of the solution set that would be selected by all parties in a blocking coalition (possibly rejecting some existing contracts). Pairwise stability means the lack of a single blocking contract. Roth showed that setwise and pairwise stable solutions coincide and exist for specific  substitutable choice functions, and a number of extensions and structural results have been obtained in the follow-up literature \cite{Roth1985}, \cite{Blair1988}, \cite{Fleiner2003},  \cite{KlausWalzl2009}, \cite{KlijnYazici2014}. 

In this paper we are focusing on the concept of strong core and Pareto-optimality under lexicographic preferences. The strong core is a classical solution concept in cooperative game theory, meaning that there is no weakly blocking coalition where there exists a matching for the coalition (without using outside contracts) that is at least as good for all of them, and strict improvement for at least one member. A solution is Pareto-optimal if the grandcoalition is not weakly blocking. It was already observed by Blair \cite{Blair1988} that the (strong) core and the set of pairwise stable solutions can be independent for many-to-many matching problems under substitutable preferences. Further examples of this kind were provided in \cite{Sotomayor1999} and \cite{KonishiUnver2006} for more restricted responsive preferences. In this paper we provide new examples for lexicographic preferences.

What is the relevance of strong core solutions in practice? The bilateral contracts in between agents can create strong bounds, even if two agents are not transacting directly, but they have a connection through third parties then they may care about the well-being of each other. In particular, one would not seek a new transaction with another agent, if this new transaction would result in a worse or terminated deal for an agent in their connected network. 

Let us consider a simple example to make this point clear. Suppose that we have four players, $a$, $b$, $c$, and $d$ and they are transacting with each other through bilateral contracts $ab$, $bc$ and $cd$. Now, $a$ and $d$ has a new potential collaboration, that would be beneficial for both $a$ and $d$, however if this happens then $a$ would cancel her partnership with $b$ making her worse off. Since $b$ and $d$ are connected through $c$, $d$ may decide not to engage in this blocking deal with $a$. 

Are such situations realistic in real world markets? Let us just substitute $a$ with Russia, $b$ with Ukraine, $c$ with USA, and $d$ with Germany. Russia is trading gas with Ukraine, but they would prefer to trade with Germany instead directly through a new channel (Nord Stream 2) and then terminate their deals with Ukraine. USA, who has a strong partnership with both Germany and Ukraine is opposing this new deal, as they are concerned about Ukraine.\footnote{The current situation with Ukraine is more complex obviously, a careful game theoretical analysis can be found about the case of Nord Stream 2 in \cite{Sziklai_etal2020}.}   

Why do we study lexicographic preferences? From a theoretical point of view this is the simplest case of preferences over bundles. When the agents are providing their strict rankings over their potential partners then lexicographic preferences over the bundles are generated in a unique, straightforward way. The responsive and the even more general substitutable preferences have a large spectrum, and a central coordinator of such a market cannot expect the agents to express their preferences over the bundles, since these can be very complex and also exponential in size. Studying the concept of (pairwise) stability can be still tractable based on the preferences over the individual partners, but for studying the (strong) core or Pareto-optimality one would need to make certain assumptions to deal with the ambiguity of possible preference extensions for bundles. \footnote{As an example, we can mention the concept of possible and necessary Pareto-optimality for responsive preferences, that was studied for allocation problems in \cite{Aziz_etal2019}. For given linear orders by the agents over individual partners, a solution is \emph{possibly Pareto-optimal} if it is Pareto-optimal for one possible responsive extension of the individual preferences, and it is \emph{necessarily Pareto-optimal} if it is Pareto-optimal for all possible responsible extension of the individual preferences.} We shall also note that our counter-examples and hardness results for lexicographic preferences are naturally valid for all the above mentioned domains, namely for additive and responsive preferences as well.

\subsection{Related literature}

Many-to-many matching markets have been studied first by Roth in \cite{Roth1984} and \cite{Roth1985}. He considered a model with multiple possible contract terms in between any worker-firm pair, from which they may select at most one. The agents at both sides select the best contracts from a possible set according to their choice functions. Roth showed that if these choice functions are substitutable then a (pairwise) stable matching always exists, and can be obtained by a deferred-acceptance algorithm. The lattice property of (pairwise) stable solutions was proved in \cite{Blair1988}, and even more general result for the existence and lattice structure were obtained by Fleiner for substitutable choice functions by using Tarski's fixpoint theorem \cite{Fleiner2003}. Klaus and Walzl studied special versions of setwise stability under different domain restrictions on substitutable preferences  \cite{KlausWalzl2009}. Klijn and Yazici proved that the rural hospitals theorem holds for substitutable and weakly separable
preferences in many-to-many markets \cite{KlijnYazici2014}.  

The efficient computation of pairwise stable solution many-to-many stable matching problems was demonstrated in \cite{BaiouBalinski2000}, and the problem of computing an optimal solution with respect to the overall rank of the matching was given in \cite{Bansal_etal2003}. For the nonbipartite stable fixtures problem Irving and Scott \cite{IrvingScott2007} provided a linear time algorithm for finding a pairwise stable solution, if one exists. Finally, Fleiner and Cechl\'arov\'a \cite{CechlarovaFleiner2005} extended these tractability results for the case of multiple contracts for the fixtures problem.

Regarding the concept of (strong) core, for many-to-one stable matching markets under responsive preferences the strong core coincides with the set of pairwise stable solutions, as shown e.g. in \cite{RothSotomayor1990}. However, for many-to-many stable matchings Sotomayor provided examples to show that the strong core and the set of pairwise stable solutions can be disjoint \cite{Sotomayor1999}. Konishi and \"Unver \cite{KonishiUnver2006} gave an example for a many-to-many stable matching problem under responsible preferences where the core is empty. (However, we shall remark that their example allowed preferences, where one agent finds another agent unacceptable alone, but bundled with another agent they together become acceptable for her.) In this paper we strengthen these results by giving an example for the emptiness of the core under the restricted domain of lexicographic preferences (where, by definition, an unacceptable agent can never be part of an acceptable bundle).

Lexicographic preferences for many-to-many assignment problems with one-sided preferences have been studied in \cite{Aziz_etal2019}, \cite{Cechlarova_etal2014}, and \cite{HosseiniLarson2019}.
However, we are not aware of any paper on lexicographic preferences for multiple partners matching problems.

\subsection{Our contribution}

First we provide an example showing that the strong core of many-to-many stable matching problems can be empty even for lexicographic preferences in Section \ref{sec:empty}. In Section \ref{sec:hardness} we prove hardness results. We show that deciding whether a many-to-many stable matching problem has non-empty core is NP-hard. We also prove that it is co-NP-complete to decide whether a given matching for a many-to-many stable matching problem is Pareto-optimal or whether it is in the strong core. We also show that finding a maximum size Pareto-optimal matching for the fixtures problem is NP-hard. On the positive side, in Section \ref{sec:easiness} we give efficient algorithms for finding a strong core solution for slightly adjusted capacities, and also for finding a half-matching that is in the strong core of fractional matchings for the stable fixtures problem. Finally, we show that finding a maximum size matching that is Pareto-optimal is possible efficiently for many-to-many problems.  

\section{Preliminaries}

First we define the \emph{two-sided many-to-many stable matching problem}, and the one-sided \emph{stable fixtures problem}.
Let $G=(N,E)$ denote the underlying graph, where the node set $N$ represents the agents and we have an undirected edge $ab\in E(G)$ if the two corresponding agents find each other mutually acceptable. Let $k(a)$ denote the (integer) capacity of agent $a$. We assume that every agent $a$ has linear preferences $>_a$ over the agents acceptable for her, where $b>_ac$ means that $a$ prefers $b$ to $c$. The solution of our problem is a \emph{matching}, that is a set of edges $M\subset E$ such that no quota is violated. If $M(a)$ denotes the set of edges incident to node $a$ in $M$ (that is the set of pair in which agent $a$ is involved in the solution), then the feasibility of the matching can be described with condition $|M(a)|\leq k(a)$ for every agent $a\in N$. If, for a matching $M$ the above condition is satisfied with equality then we say that the agent is \emph{saturated}, otherwise she is \emph{unsaturated}. When $G$ in non-bipartite then we get the stable fixtures problem \cite{IrvingScott2007}, and when $G$ is bipartite then we get the many-to-many stable matching problem, see e.g. \cite{BaiouBalinski2000}. 

The classical solution concept for these problems is (pairwise) stability. A matching $M$ is \emph{stable} if there is no blocking pair. A pair $ab\notin M$ is \emph{blocking}, if $a$ is either unsaturated, or there is $ac\in M$ such that $b>_a c$, and likewise, $b$ is either unsaturated or there is $bd\in M$ such that $a>_b d$. 

When all the capacities are unit then for two-sided problems we get the \emph{stable marriage problem}, and the one-sided case is called \emph{stable roommates problem}, as defined by Gale and Shapley \cite{GaleShapley1962}. Gale and Shapley gave an efficient algorithm for finding a stable matching for the marriage case, and demonstrated with an example that stable matching may not exists for the roommates case. Irving \cite{Irving1985} gave a linear time algorithm that can find a stable solution for the roommates problem, if one exists. The results are similar for the capacitated case, a stable solution always exists for two-sided problems and can be computed in linear time by a generalised Gale-Shapley type algorithm, see e.g. \cite{BaiouBalinski2000}. For the stable fixtures Irving and Scott \cite{IrvingScott2007} provided a linear time algorithm to find a stable solution, if one exists.     

In this paper we focus on the (strong) core and Pareto-optimality of the solutions, so we need to extend the preferences of the agents over the set of partners. Let $\succ_a$ denote the linear preferences of agent $a$ over the possible set of partners. We will assume that the preferences of the agents are lexicographic in the sense that they mostly care about their best partner, and then about their second best partner, and so on. 
Formally, we define the preference relation $\succ_a$ for agent $a$ over the sets $S,T\subset N$ the following way. Consider the characteristic vector of $S $ and $T$, denoted by $\chi_S$ and $\chi_T$, where the order of the coordinates are the same as the preference order of $a$ over the elements in $N$. Then $S\succ_a T$ if and only if $\chi_S$ is lexicographically greater than $\chi_T$. This definition can be easily extended to the fractional case, to be defined in Section \ref{sec:easiness}. Note that lexicographic preferences are strict by definition, so if $S\succeq_a T$ then either $S\succ_a T$ or $S=T$. 

A matching $M$ is in the \emph{core}, if there is no \emph{blocking coalition} $S$ with an alternative matching $M'$ on $S$ that is strictly preferred by all the members of $S$, that is $M'(a)\succ_a M(a)$ for every $a\in S$. A matching $M$ is in the \emph{strong core}, if there is no \emph{weakly blocking coalition} $S$ with alternative matching $M'$ on $S$ that is weakly preferred by all the members and strictly preferred by at least one member in $S$\footnote{Note that this means that in a weakly blocking coalition for lexicographic preferences everybody either gets strictly better partners or gets the same partners as before.}. Finally, a matching is \emph{Pareto-optimal} if it is not weakly blocked by $N$.

\subsection*{Examples for stability versus core property}\label{examples}

Here we provide two examples to demonstrate the differences in between stable matchings, strong core and Pareto-optimal matchings.

\subsubsection*{Example 1}
We have four agents on both sides, $A=\{a,b,c,d\}$ and $B=\{x,y,z,w\}$ having capacity two each, and with the following linear preferences on their potential partners:\\

\begin{center}
\begin{tabular}{rl|rl}
$a:$ & $x > z > w > y$ & $x:$ & $b > c > d > a$\\
$b:$ & $y > z > w > x$ & $y:$ & $a > c > d > b$\\
$c:$ & $x > y$ & $z:$ & $a > b$\\
$d:$ & $x > y$ & $w:$ & $a > b$\\
\end{tabular}
\end{center}

Here, the unique pairwise stable solution is $M=\{az,aw,bz,bw,cx,cy,dx,dy\}$, and the unique strong core solution is $M'=\{ax,ay,bx,by\}$ when we assume that agents have lexicographic preferences. Note that both of these solutions are Pareto-optimal.

The next, extended example shows that the unique stable solution may not even be Pareto-optimal.

\subsubsection*{Example 2}
We have five agents on both sides, $A=\{a,b,c,d,p\}$ and $B=\{x,y,z,w,q\}$ having capacity two each, and with the following linear preferences on their potential partners:\\

\begin{center}
\begin{tabular}{rl|rl}
$a:$ & $x > y > z > q > w$ & $x$: & $d > c > b > p > a$\\
$b:$ & $y > x > w > q > z$ & $y:$ & $c > d > a > p > b$\\
$c:$ & $z > w > x > q > y$ & $z:$ & $b > a > d > p > c$\\
$d:$ & $w > z > y > q > x$ & $w:$ & $a > b > c > p > d$\\
$p:$ & $x > y > z > w > q$ & $q:$ & $a > b > c > d > p$\\
\end{tabular}
\end{center}

Here, the unique pairwise stable solution is $M=\{ay, az, bx, bw, cw, cx, dz, dy, pq\}$, and the unique strong core solution is $M'=\{ax, aw, by, bz, cz, cy, dw, dx, pq\}$. Note that $M'$ also Pareto-dominates $M$, so no stable matching is Pareto-optimal for this example.\\

\subsubsection*{Example 3}

We have a stable fixtures problem with ten agents and the following preferences:
\begin{center}
\begin{tabular}{rl}
    $x_1:$ & $x_2>x_4>x_3$ \\
    $x_2:$ & $x_1>x_5>x_6$ \\
    $x_3:$ & $x_7>x_1$ \\
    $x_4:$ & $x_8>x_1$ \\
    $x_5:$ & $x_9>x_2$ \\
    $x_6:$ & $x_{10}>x_2$ \\
    $x_7:$ & $x_3>x_8$ \\
    $x_8:$ & $x_4>x_7$ \\
    $x_9:$ & $x_5>x_{10}$ \\
    $x_{10}:$ & $x_6>x_9$\\
\end{tabular}
\end{center}
The capacities of agents $x_1$ and $x_2$ are 2, the capacities of the others are 1. Here the only complete matching is $M=\{ x_1x_3,x_1x_4,x_2x_5,x_2x_6, x_7x_8, x_9x_{10} \}$, but the matching $M'=\{ x_1x_2, x_3x_7, x_4x_8, x_5x_9, x_6x_{10} \}$ Pareto-dominates it, so there is no complete Pareto-optimal matching in this instance. We will see in Section \ref{sec:easiness} that for a many-to-many stable matching problem a maximum size Pareto-optimal matching always exists and one can be found in polynomial time, but the same problem is NP-hard for the fixtures problem. 
\begin{figure}
    \centering
    \includegraphics[height=0.15\textheight]{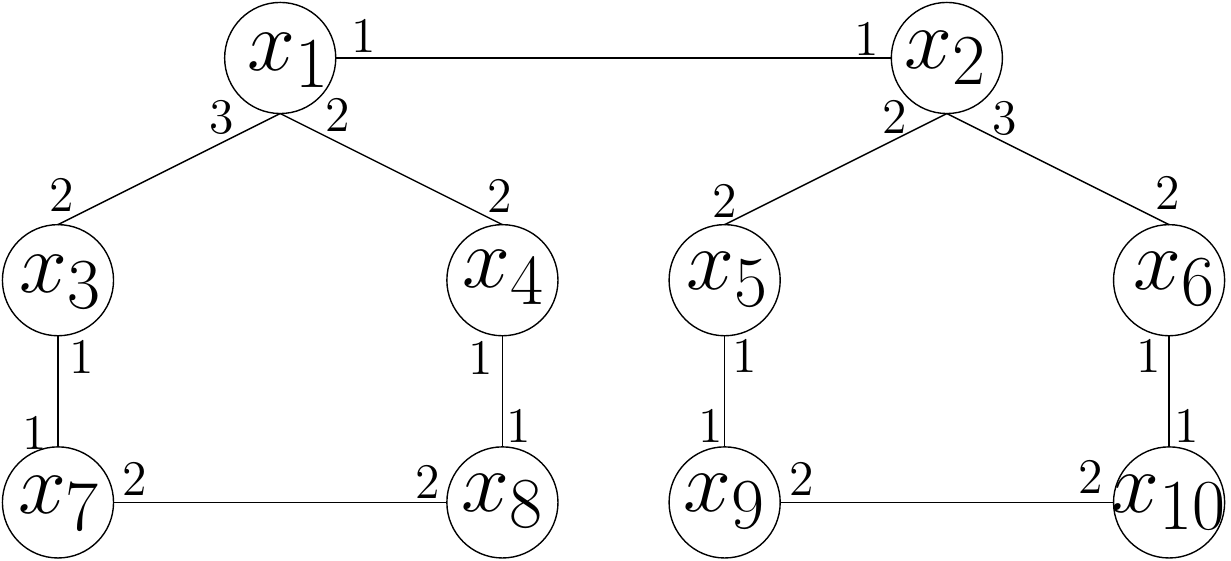}
    \caption{Example 3}
    \label{pareto}
\end{figure}

\section{Emptiness of the strong core}\label{sec:empty}

In this section we show that the strong core of a many-to-many stable matching problem may be empty, even under lexicographic preferences.

\begin{Th}
\label{ell}
The strong core of many-to-many matching markets may be empty, even if the preferences are lexicographic.
\end{Th}
\begin{proof}

We construct an instance where the strong core is empty. The construction is the following: There are 12 agents, on one side we have agents $a$ and $b$ with capacity $2$, $c,d,x',y'$ with capacity one, and on the other side we have agents $x,y$ with capacity $2$ and $u,v,a',b'$ with capacity one. The preferences of the agents are shown as follows:
\begin{center}
\begin{tabular}{rl|rl}
$a:$ & $u > y > v > a'>x$ & $x:$ & $d > a > c > x'>b$\\
$b:$ & $v > x > u > b'>y$ & $y:$ & $c > b > d > y'>a$\\
$c:$ & $x > y$ & $u:$ & $b>a $\\
$d:$ & $y > x$ & $v:$ & $a> b$\\
$x':$ & $x$ & $a':$ & $a$\\
$y':$ & $y$ & $b':$ & $b$\\
\end{tabular}
\end{center}

\begin{figure}
\centering
\label{fig1}

\includegraphics[width=6cm]{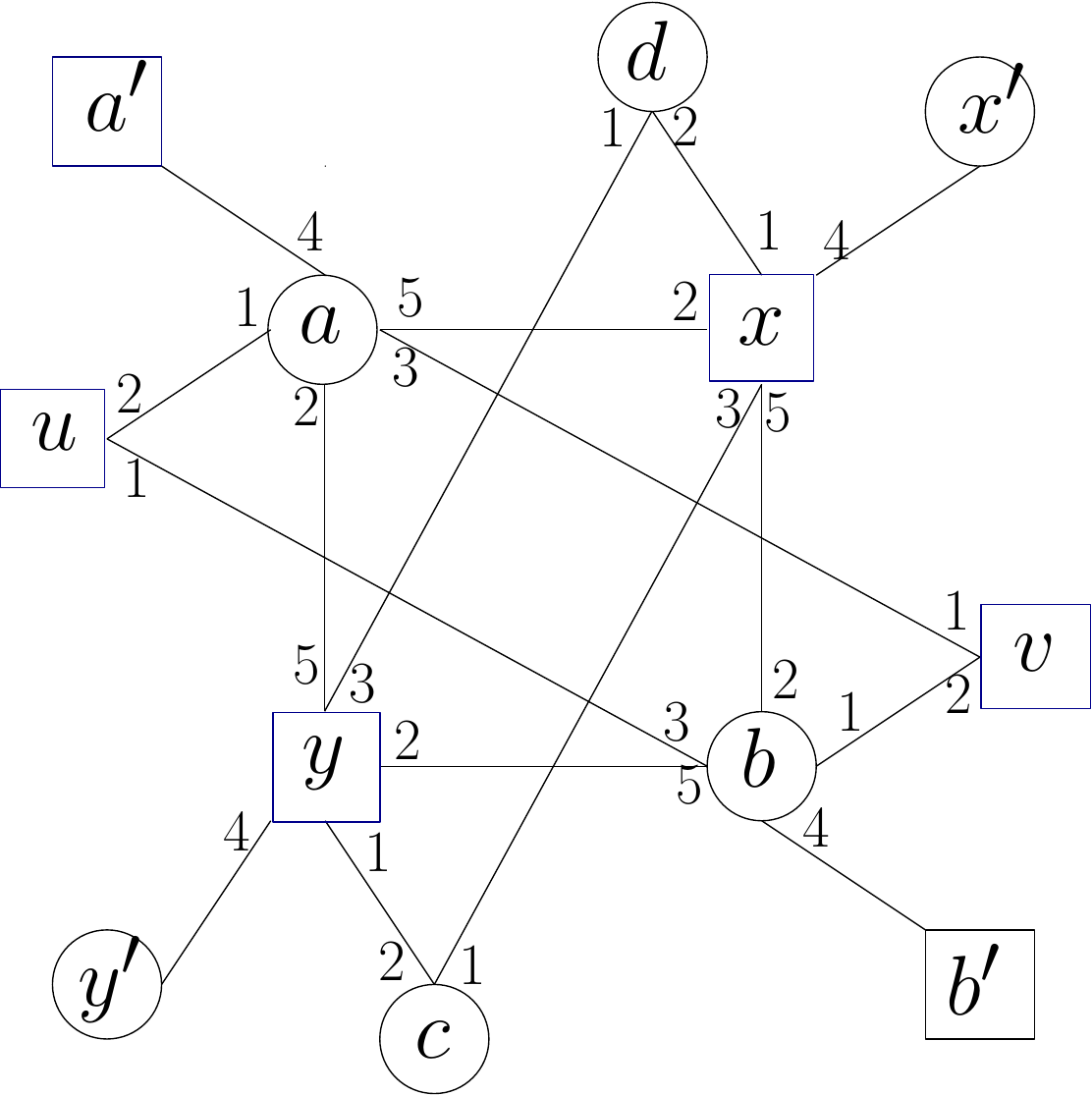}
\caption{An illustrative image of the counterexample used in Theorem \ref{ell}}
\end{figure}
Let us suppose that there is a matching $M$ in the strong core. Clearly, if any one of $\{ c,d,u,v \}$ is unmatched, then they form a blocking coalition with their second choice, since they are their second choice's best option. This also means that the middle four cycle $C=\{ ax,xb,by,ya\}$ cannot be included either, since they are the only possible partners of $\{ c,d,u,v\}$. This also means that any possible matching $M$ in the strong core has to be acyclic.

Observe that $a,b,x,y$ must be saturated (otherwise they would block with their dummy partner $a',b',x'$ or $y'$ and the rest of the acyclic component containing them.

Suppose that there is an edge of the four cycle $C$ that is included in $M$, suppose by symmetry it is $ax$. Then $aa'$ would block it with the rest of $a$'s component in $M$, because $a$ and $x$ cannot be connected with a path. So no edges of $C$ can be in $M$.

If any of $\{au, bv, xd, yc\}$ is included in the matching, then $bu$, $av$, $yc$ or $xc$ would block with the rest of the given component, respectively.  
This means that the only possible choice left for $M$ is $\{ av,aa',bu,bb',yd,yy',xc,xx'\}$, but then the four cycle $C$ in the middle would block, a contradiction.
So the strong core of the instance is indeed empty.

\end{proof}

\section{Hardness results}\label{sec:hardness}

In this section we prove NP-hardness results.

\subsection{Deciding the non-emptiness of the strong core}

\begin{Th}
\label{np-core}
Deciding whether the strong core of a many-to-many stable matching problem is non-empty is NP-hard under lexicographic preferences, even if each capacity is at most two. 
\end{Th}
\begin{proof}
We reduce from an NP-complete special version of the \textsc{com-smti} problem, which was shown to be NP-hard by Manlove et al.\cite{Manlove_etal2002}. This problem is a special instance of the stable marriage problem with ties and incomplete preference lists such that there are no ties in the preferences of the men $U=\{ u_1,..,u_n\}$ and the set of woman can be partitioned into two parts $W=W^s\cup W^t=\{ w_1,..,w_k\} \cup \{ w_{k+1},...,w_n\}$ such that the preference lists of the woman in $W^s$ have no ties and the preference lists of the woman in $W^t$ consists of only a single tie. The task is to find a complete weakly stable matching, that is a matching $M$ that pairs every man and woman with someone and there is no pair $(m,w)\notin M$ such that they both strictly prefer each other to their partner in $M$.

\begin{figure}
    \centering
    \includegraphics[height=0.2\textheight]{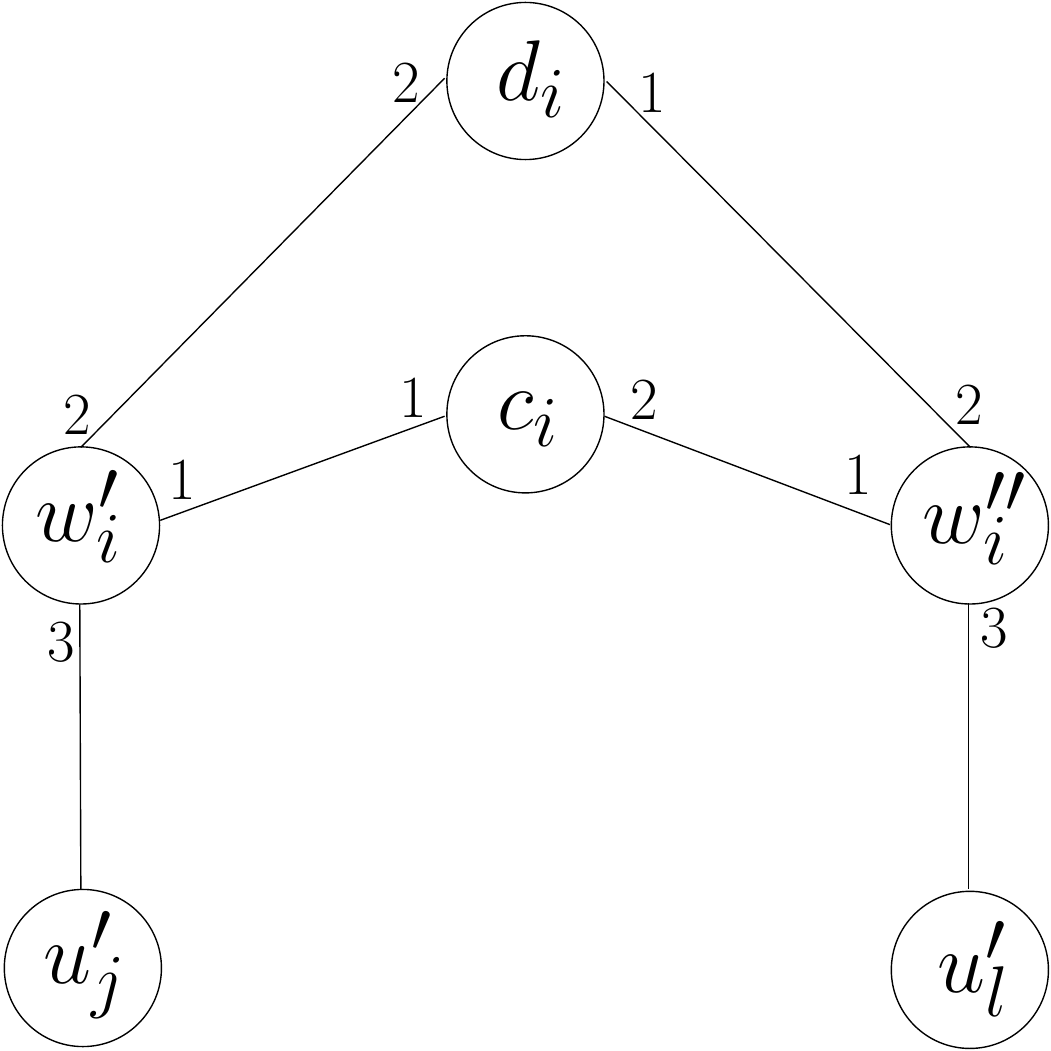}
    \caption{The gadget $G_i$ for theorem \ref{np-core} if the tie in $w_i$'s list was $u_j\sim u_l$ }
    \label{woman}
\end{figure}
Suppose that we have an instance $I$ of the above problem. We construct an instance $I'$ of the multiple partner stable matching problem such that the strong core is nonempty if and only if there is a complete weakly stable matching in $I$.

First of all, for every man $u_i\in U$ we create a vertex $u'_i$ with capacity one, and for every woman $w_i\in W^s$ we create a vertex $w_i'$ with capacity one. Denote these two sets by $U'$ and $W^{s'}$. Then, for every woman $w_i\in W^t$ we create a gadget $G_i$ illustrated in Figure \ref{woman}. We create four vertices: $w_i',w_i''$ and $c_i$ having capacity two and $d_i$ having capacity one. $w_i'$ is connected to one man in $w_i$'s preference list and $w_i''$ is connected to the other. The preferences are the following:

\begin{center}
\begin{tabular}{rl}
$c_i:$ & $w_i'>w_i''$ \\
$d_i:$ & $w_i''>w_i'$ \\
$w_i':$ & $c_i>d_i>u_j'$ \\
$w_i'':$ & $c_i>d_i>u_l'$ \\
\end{tabular}
\end{center}
Finally, we add a gadget $G$, which is just a copy of the counterexample from Figure \ref{ell} and a special agent $g$ with capacity one.

The preferences of the agents in $U'$ are the same just over the agents $w_i'$ instead of $w_i$, except if their was a woman $w_i\in W^t$ in their preference list, then we substitute $w_i$ with the appropriate copy from $\{ w_i',w_i''\}$. Finally, we add the special agent $g$ to the end of all of their preference lists.

Similarly, for each $w_i'\in W^{s'}$, the preference lists are the same with $u_i'$-s instead of $u_i$-s.
The agents in the gadget $G$ have the same  preferences, except that we add $g$ to the beginning of $a\in G$'s list.
Finally, the preference lists of $g$ has the agents in $U'$ in arbitrary order followed by $a\in G$ in the end.

Now let us suppose that we have a complete weakly stable matching $M$ in $I$. We create a matching $M'$ in $I'$ by adding an edge $u_i'w_j'$ or $u_i'w_j''$ (the one which exists) for each $u_iw_j\in M$. Also for each gadget $G_i$ we add the edges $w_i'c_i$ and $w_i''c_i$ to $M'$. If the partner of $w_i\in W^s$ was $u_j$, then we add $w_i''d_i$, if it was $u_l$, then we add $w_i'd_i$. Finally, we add the edges  $ag$ and $\{ av,bu,yd,yy',xc,xx', bb'\}$ to $M'$. 

We show that $M'$ is in the strong core. Let us suppose there is a blocking coalition $C$ for $M'$. If there is a vertex $v_i\in \{ w_i',w_i'',c_i,d_i\}$ from a gadget $G_i$ in $C$, then all of them are in $C$, since if $v_i=d_i$, then $w_i'$ or $w_i''\in C$, so their favorite partner $c_i$ is also in $C$ and so is the other copy of $w_i$. Similarly if $w_i'$ or $w_i''$ or $c_i\in C$, then all of them are in $C$ and the two copies of $w_i$ get the same partner, so none of them can achieve a strictly better situation. So no agents from $G_i$ can improve their situation.

If a man $u_i'$ is strictly better off in $C$, then she has to have a better partner and also she has to be at least as good a partner for her. She cannot be strictly better, since $M$ was weakly stable. So the partner $w_i$ has to be from $W^t$. But then, the corresponding copy of $w_i$ was matched to $c_i$ and $d_i$ in $M'$, both of which it prefers to $u_i'$, so they cannot be paired in a blocking coalition, a contradiction. 

If an agent $w_i'\in W^{s'}$ gets a strictly better partner in $C$, then she and her partner would form a blocking pair to $M$, a contradiction.

Special agent $g$ cannot get a better partner in $C$, because she is the worst choice for every other possible partner other than $a$, and every one of them is at full capacity, since $M$ was a complete matching. 

Finally, it is easy to check, that there are no blocking coalition in the gadget $G$ to $M'$ either, so $M'$ is in the strong core.

For the other direction suppose that $M'$ is in the strong core of $I'$. This implies that $ga\in M'$, since there is no strong core solution among the agents in $G$. Therefore every agent in $U'$ must be matched to someone in $W'$, because otherwise they would block with $g$. 

Now, we create $M$ the following way: for each $u_i'\in U'$ we assign $u_i$ the woman corresponding to the partner of $u_i'$ in $M'$. To see that no two man gets the same partner, suppose that $u_j$ and $u_l$ does. Then, $u_j'w_i'$ and $u_l'w_i''$ are both in $M'$ for a woman in $W^{t'}$. But this means that $c_i$ and $d_i$ cannot be both saturated, so one of them blocks with $g$ and the rest of its (acyclic) component in $M'$.

Since every man $u_i$ is matched and to different partners, it follows that $M$ is a complete matching.

Now suppose there is a strictly blocking pair $(u_i,w_i)$. Then $w_i\in W^s$ and $\{ u_i',w_i'\}$ would form a blocking coalition for $M'$, a contradiction.
 
So $M$ is a complete, and weakly stable matching.
\end{proof}

\subsection{Checking Pareto-optimality of a matching}

First we show that checking Pareto-optimality is co-NP-complete, and then we prove that this implies a similar hardness result for the problem of checking the strong core property. We say that a matching \emph{complete} if every vertex is saturated

\begin{Th}
\label{pareto1}
Deciding whether a given matching is Pareto-optimal for the many-to-many stable matching problem under lexicographic preferences is co-NP-complete, even for complete matchings.
\end{Th}

\begin{proof}
The problem is in co-NP, since checking that an alternative matching $M'$ Pareto-dominates $M$ can be done efficiently. We reduce from {\sc Exact-3-Cover}, where we are given a set of $3n$ items $X=\{x_1, x_2, \dots , x_{3n}\}$ and a set of $m$ 3-sets, $\mathcal{Y}=\{Y_1, Y_2, \dots ,Y_m\}$, where each $Y_j$ contains 3 items from $X$ and the decision question is whether there exists a subset $\mathcal{Y'}\subset\mathcal{Y}$ of size $n$ that admits all the elements of $X$ exactly ones. Given an instance $I$ of {\sc Exact-3-Cover}, as described above, we create an instance $I'$ of many-to-many stable matching problem as follows. We will have five gadgets, each with two sets of agents, $A\cup B$, $C\cup D$, $P\cup Q$, $S\cup T$ and $U\cup V$. More specifically, the set of agents are as follows.

\begin{center}
\begin{tabular}{l|l}
$A=\{a_1, a_2, \dots , a_{3n}\}$ & $B=\{b_1, b_2, \dots , b_{3n}\}$\\
$C=\{c_1, c_2, \dots , c_{n}\}$ & $D=\{d_1, d_2, \dots , d_{n}\}$\\
$P=\{p_1, p_2, \dots , p_{3n}\}$ & $Q=\{q_1, q_2, \dots , q_{3n}\}$\\
$S=\{s_1, s_2, \dots , s_{m}\}$ & $T=\{t_1, t_2, \dots , t_{m}\}$\\
$U=\{u_1, u_2, \dots , u_{4m}\}$ & $V=\{v_1, v_2, \dots , v_{4m}\}$\\
\end{tabular}
\end{center}

Let the capacity of every agent be 2 in $A\cup B$, 3 in $C\cup D$, 1 in $P\cup Q$, 4 in $S\cup T$ and 1 in $U\cup V$. Finally, we describe the linear orders of the agents on their acceptable partners. Here, for any set $E$, $[E]$ denotes the elements of $E$ in the order of the elements' indices. Furthermore, we define $Q_j=\cup\{q_i:x_i\in Y_j\}$ and $S_i=\cup\{s_j:x_i\in Y_j\}$, and similarly $P_j=\cup\{p_i:x_i\in Y_j\}$ and $T_i=\cup\{t_j:x_i\in Y_j\}$.

\footnotesize
\begin{center}
\begin{tabular}{rl|rl}
$a_1:$ & $b_1> q_1> d_{1}> b_{3n}$  & $b_1:$ & $a_2> p_1> c_{1}> a_{1}$ \\
 & $\dots $ & $\dots $ & \\
$a_i:$ & $b_i> q_i> d_{\lfloor (i+2)/3\rfloor}> b_{i-1}$  & $b_i:$ & $a_{i+1}> p_i> c_{\lfloor (i+2)/3\rfloor}> a_{i}$ \\
 & $\dots $ & $\dots $ & \\
$a_{3n}:$ & $b_{3n}> q_{3n}> d_{n}> b_{3n-1}$  & $b_{3n}:$ & $a_1> p_{3n}> c_{n}> a_{3n}$\\
\hline
$c_1:$ & $d_1> b_1> b_2> b_3> d_n> [T]$  & $d_1:$ & $c_2> a_1> a_2> a_3> c_1> [S]$ \\
 & $\dots $ & $\dots $ & \\
$c_i:$ & $d_i> b_{3i-2}> b_{3i-1}> b_{3i}> d_{i-1}> [T]$  & $d_i:$ & $c_{i+1}> a_{3i-2}> a_{3i-1}> a_{3i}> c_i> [S]$ \\
 & $\dots $ & $\dots $ & \\
$c_n:$ & $d_n> b_{3n-2}> b_{3n-1}> b_{3n}> d_{n-1}> [T]$  & $d_n:$ & $c_{1}> a_{3n-2}> a_{3(n-1)+2}> a_{3n}> c_n> [S]$ \\
\hline
$p_1:$ & $[T_1]> b_1$  & $q_1:$ & $[S_1]> a_1$ \\
 & $\dots $ & $\dots $ & \\
$p_i:$ & $[T_i]> b_i$  & $q_i:$ & $[S_i]> a_i$ \\
 & $\dots $ & $\dots $ & \\
$p_n:$ & $[T_n]> b_n$  & $q_n:$ & $[S_n]> a_n$ \\
\hline
$s_1:$ & $[D]> v_1> v_2> v_3> v_4> [Q_1]$  & $t_1:$ & $[C]> u_1> u_2> u_3> u_4> [P_1]$ \\
 & $\dots $ & $\dots $ & \\
$s_j:$ & $[D]> v_{4j-3}> v_{4j-2}> v_{4j-1}> v_{4j}> [Q_j]$  & $t_j:$ & $[C]> u_{4j-3}> u_{4j-2}> u_{4j-1}> u_{4j}> [P_j]$ \\
 & $\dots $ & $\dots $ & \\
$s_m:$ & $[D]> v_{4m-3}> v_{4m-2}> v_{4m-1}> v_{4m}> [Q_n]$  & $t_m:$ & $[C]> u_{4m-3}> u_{4m-2}> u_{4m-1}> u_{4m}> [P_m]$ \\
\hline
 $u_1:$ & $v_1> t_1$  & $v_1:$ & $u_1> s_1$ \\
& $\dots $ & $\dots $ & \\
 $u_j:$ & $v_j> t_{\lfloor (j+3)/4\rfloor}$  & $v_j:$ & $u_j> s_{\lfloor (j+3)/4\rfloor}$ \\
& $\dots $ & $\dots $ & \\
 $u_{4m}:$ & $v_{4m}> t_{m}$  & $v_{4m}:$ & $u_{4m}> s_m$ \\
\end{tabular}
\end{center}
\normalsize

We create a matching $M$ in $I'$ as follows. Let each agent in $A$ and $B$ be matched with their acceptable partners in $D\cup Q$ and in $C\cup P$, respectively. Furthermore, each agent in $S$ is matched with all of her four acceptable agents in $V$, similarly, each agent in $T$ is matched with all of his four acceptable agents in $U$.  We depict the accessibility graph of $I'$ in Figure \ref{lex} with regard to the main sets of agents, where the solid edges mark that all of the  mutually acceptable pairs between the two corresponding sets belong to $M$ and the dashed edges denote when no edge between the corresponding sets belongs to $M$.

\begin{figure}
\begin{center}
\scalebox{0.4}{\includegraphics{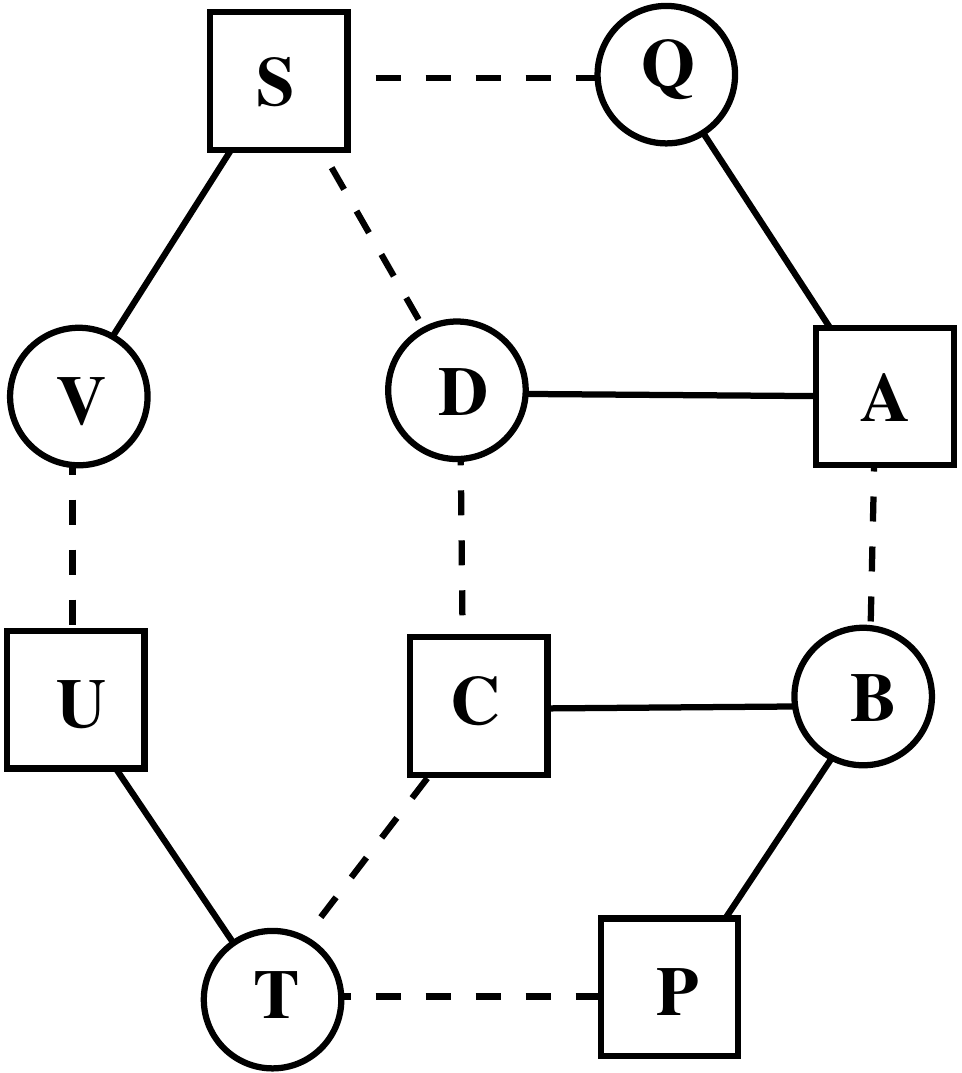}}
\caption{The mutually acceptable pairs between the main sets of agents, where the solid edges denote the projection of $M$}\label{lex}
\end{center}
\end{figure}

Now we shall prove that $I$ has an exact 3-cover if an only if matching $M$ is not Pareto-optimal in $I'$.

First, let us suppose that we have an exact 3-cover $\mathcal{Y'}$ in $I$. We create a matching $M'$ in $I'$ that Pareto dominates $M$ in the following way. In $M'$ we match each agent in $A$ to her two acceptable partners in $B$, which implies that we also match each agent in $B$ to his two acceptable partners in $A$. Likewise, we match each agent in $C$ to her two acceptable partners in $D$, which implies that we also match each agent in $B$ to his two acceptable partners in $A$. For the rest of the agents we create $M'$ according to the 3-cover $\mathcal{Y'}$ as follows. If $Y_j\in \mathcal{Y'}$ then we match $s_j$ to the three agents in $Q_j$ and also to an arbitrary agent in $D$, and similarly, we match $t_j$ to the three agents in $P_j$ and also to an arbitrary agent in $C$, and finally we also match $u_{4(j-1)+k}$ to $v_{4(j-1)+k}$ for every $k\in\{1, 2, 3, 4\}$. For those agents in $s_j\in S$ and $t_j\in T$, where $Y_j\notin Y'$, we keep the edges of $M$. It is easy to see that all the agents that changed partners in $I'$ improved according to their lexicographic preferences, since all of them become matched to their best potential partner in $M'$.

In the other direction, let us suppose that $M$ is not-Pareto optimal, so there is an alternative matching $M^*$ that Pareto-dominates it. We shall prove that $I$ has an exact 3-cover. First we note that if any agent in $A\cup B\cup C\cup D$ has a different partner in $M^*$ than in $M$ (thus necessarily improves) then the matching between sets $A$ and $B$ and also between $C$ and $D$ must be complete, as we had in $M'$. This also implies that all the agents in both $P$ and $Q$ must get new partners in $M^*$ from the sets $T$ and $S$, respectively. However, this is only possible if at least $n$ agents form both $S$ and $T$ get also new partners from the sets $D$ and $C$, respectively. But the agents in $C$ and $D$ have remaining capacity one each, so they should become matched with exactly $n$ agents from each of $T$ and $S$, respectively, so these are the only $2n$ agents from these sets that can change partners and help improve to those in $P$ and $Q$. To summarise, if these agents all improve then we must be able to choose an exact 3-cover by adding $Y_j$ to $\mathcal{Y'}$ if $s_j$ has improved in $M^*$. What remained is to show that the improvement of any other agent outside the set $A\cup B\cup C\cup D$ would also lead to the same effect. Indeed, if any agent in $U\cup V$ improves in $M^*$ then her/his partner in $M$ from $S\cup T$ must also improve and this is only possible is the latter agent gets matched with someone from $A\cup B\cup C\cup D$. The same applies if any agent in $P\cup Q$ would improve. Thus we can conclude that the improvement of any agent in $I'$ implies that all agents in $A\cup B\cup C\cup D$ must improve and thus we are able to find an exact 3-cover in $I$.
\end{proof}

\subsection{Finding a complete Pareto-optimal matching}

In this section we show that in the stable fixtures case, the problem of finding a maximum size Pareto-optimal matching is NP-hard. We prove this for complete matchings.

\begin{Th}
Deciding whether there exists a complete Pareto-optimal matching for the stable fixtures problem under lexicographic preferences is NP-hard, even if each capacity is at most 4.

\end{Th}
\begin{proof}
Again we reduce from the problem \textsc{exact-3-cover}. We use almost the same construction as in theorem \ref{pareto1} with the only difference that here we substitute each $a_ib_i$ and $a_ib_{i-1}$ edge with a gadget $G_i$ and $H_i$ respectively. Every gadget $G_i$ and $H_i$ are essentially just a copy of Example 3 in Section 2, illustrated in Figure \ref{pareto}, only we add a special agent $g_i$ and $h_i$ respectively. We will denote the agents corresponding to $x_1,...,x_{10}$ in $G_i$ by $x_1^i,...,x_{10}^i$ and in $H_i$ by $y_1^i,...,y_{10}^i$. An agent $g_i$ has preference $b_i>x_7^i>x_8^i>a_i$ and is added to the end of the preference lists of both $x_7^i$ and $x_8^i$. An agent $h_i$ has preference $a_i>y_7^i>y_8^i>b_{i-1}$ and is added to the end of the preference lists of both $y_7^i$ and $y_8^i$. Finally, we substitute $b_i$ and $b_{i-1}$ in the preference list of $a_i$ by $g_i$ and $h_i$ respectively and similarly substitute $a_i$ and $a_{i+1}$ in $b_i$'s preference list by $g_i$ and $h_{i+1}$ respectively for each $i=1,...,3n.$

Suppose that there is a complete Pareto-optimal matching $M$ in this instance and suppose that there is an index $i$ such that $a_ig_i$ and $g_ib_i$ are in $M$. Then the matching $M$ restricted to the se of vertices $\{x_1^i,...,x_{10}^i\}$ has to be $\{ x_1^ix_3^i, x_1^ix_4^i, x_2^ix_5^i, x_2^ix_6^i, x_7^ix_8^i, x_9^ix_{10}^i \}$ by the completeness of $M$, but then if we give each agent apart from $x_1^i,..,x_{10}^i$ the same partners and match $x_1^i,...x_{10}^i$ such that each gets only their favourite, then we obtain a matching $M'$ that Pareto-dominates $M$, a contradiction. 

Similarly, if $a_ih_i$ and $h_ib_{i-1}$ are in $M$ for some $i$, then $M$ cannot be Pareto-optimal either, a contradiction.

Since a gadget $G_i$ or $H_i$ can be saturated only if $g_i$ or $h_i$ is matched to 0 or 2 agents in them, we obtain that there can be no edge in $M$ that connects an agent $a_i$ or $b_i$ to an agent $g_j$ or $h_j$. But $M$ is a complete matching, hence each agent in $A\cup B$ is saturated, so all edges between $A$ and $Q\cup D$ has to be included and also all edges between $B$ and $C\cup P$. Then, since every agent in $T\cup S$ is saturated too, all edges between $U$ and $T$ and all edges between $S$ and $V$ are also included $M$, so $M$ is basically the same matching that we constructed in Theorem \ref{pareto1}. 

Now if there would be an exact 3-cover in $I$, then we could construct a matching $M'$ that Pareto-dominates $M$, implicating that there can be no complete Pareto-optimal matching the same way as before, with the addition that the agents in $A\cup B$ obtain their partners in $\cup_ig_i \cup \cup_ih_i$ instead of each other. This way each agent $g_j$ and $h_j$ obtains their best partner, so they are strictly better off, too. Finally, we also match each agent in the remaining parts of  $\cup_iG_i \cup \cup_iH_i$ to their best choices. So the existence of an exact 3-cover implicates that no complete Pareto-optimal matching exists. 

In the other direction if there is no complete Pareto-optimal matching, then the matching $M$ constructed above is not Pareto-optimal, so it is dominated by a matching $M'$. Again, the same proof works to show that there has to be a 3-cover of the original instance. The only additional thing we have to check in this case is that, if any agent from a gadget $G_i$ or $H_i$ improves their position, then so does every agent in $A\cup B\cup C\cup D$. But this is only possible if it gets its first choice in $M'$, which implicates that every agent in the gadget obtains its best partner, so $g_i$ or $h_i$ improves their position too, which leads to every edge between $A\cup B$ and $\cup_ig_i\cup \cup_ih_i$ included in $M'$, so the proof is completed.
\end{proof}

\subsection{Checking the strong core property of a matching}

\begin{Th}
Deciding whether a given matching is in the strong core of a many-to-many stable matching problem under lexicographic preferences is co-NP-complete.
\end{Th}

\begin{proof}
The problem is in co-NP, since checking that $M$ can be blocked by a coalition $S$ with an alternative matching $M_S$ can be done efficiently. We reduce from the problem of checking Pareto-optimality for many-to-many stable matching problem under lexicographic preferences, that we showed to be co-NP-complete in Theorem \ref{pareto1}. Suppose that we have such an instance $I$, thus a many-to-many market with linear orders of the agents, and a matching $M$ that is to be checked to be Pareto-optimal. We add two new agents $a^*$ and $b^*$ to the market, one to each side, such that they have unbounded capacities and they find everyone acceptable and most importantly everyone in $I$ put either $a^*$ or $b^*$ as her top choice in $I'$. Let us also increase the capacity of all the agents in $I$ by one and then we create $M'$ as an extension of $M$ in the following way, we add $a^*b^*$ and we match all the agents in $I$ to either $a^*$ or $b^*$ (according to the side they belong to). Now, it is easy to see that $M$ is Pareto-optimal in $I$ if and only if $M'$ is in the strong core in $I'$, since a blocking coalition in $I'$ must involve every agent in $I'$ and the alternative matching of the blocking coalition must keep all the pairs in $M'\setminus M$.
\end{proof}

\section{Tractable cases}\label{sec:easiness}

In this section we show that although many natural problems related to strong core and Pareto-optimality are NP-hard, we can still find some reasonable solutions efficiently.

\subsection{Near feasible and fractional matchings with TTC}

We give two algorithms, that are heavily inspired by the Top Trading Cycle (TTC) algorithm of Gale \cite{ShapleyScarf1974}. One computes a matching $M$, such that $M$ violates the original capacity constraints by at most one, but is guaranteed to be a strong core solution for this slightly modified instance. We call such a solution a \textit{near-feasible strong core solution}.
The other computes a half-matching $M$, that is guaranteed to be in the strong core of fractional solutions of the original instance. 

We define fractional matchings as $M:E\to [0,1]$ functions, such that $\sum_{u:uv\in E}M(uv)\le k(v)$ for every $v\in V$. Now, if $S$ and $T$ are fractional matchings, then $S\succ_a T$ if and only if the fractional characteristic vectors (so $\chi_S$ has $M(av)$ at its $v$ coordinate) satisfy $\chi_S>\chi_T$ lexicographically.

The algorithms described here not only work for the many-to-many stable matching case, but also for the non-bipartite stable fixtures problem. Moreover, both algorithms run in linear time in the number of edges.

The main idea of the algorithms is very simple: in each step, we create a directed graph $D_i=(V_i,A_i)$, such that the vertices of $D_i$ are the agents who have remaining capacities at the $i$-th iteration, and there is a directed edge from $a$ to $b$ if $b$ is $a$'s best choice from the vertices of $D_i$ who are not yet matched to $a$.  Then we search for a directed cycle $C_i$ in $D_i$ and add the edges of $C_i$ to the matching. 

Now we describe the algorithms formally. Let $p_U^M(v)$ denote the best agent in $v$'s preference list among the agents in $U$ who are not matched to $v$ in $M$. Let $k(v)$ denote the capacity of $v$. Also we use $E(C_i)$ as the edges corresponding the the directed edges of $A(C_i) $ in the original graph $G$.
When we use the notation $\frac{1}{2}e$, we mean that we only add the edge $e$ with $\frac{1}{2}$ weight to $M$.

\begin{algorithm}\caption{Near-feasible strong core}
\begin{algorithmic}
\State Set $M=\emptyset$
\State $V_0=N$, $A_0=\{ vp_V^M(v) :\; v\in V_0\}$,
\While{$A_i\ne \emptyset$}
\State Find a directed cycle $C_i$ in $D_i$. 
\State For each $e\in E(C_i)$: $M:=M\cup e$.
\If{$|C_i|=2$}
\State For each $v\in V(C_i):$ $k(v)=k(v)-1$
\Else
\State For each $v\in V(C_i):$ $k(v)=k(v)-2$
\EndIf
\State $V_{i+1}=\{ v\in V:\; k(v)\ge 1\}$
\State $A_{i+1}=\{ vp_{V_{i+1}}^M(v):\; v\in V_{i+1}\}$
\EndWhile
\end{algorithmic}
\end{algorithm}

\begin{algorithm}\caption{Half-integer strong core}
\begin{algorithmic}
\State Set $M=\emptyset$
\State $V_0=N$, $A_0=\{ vp_V^M(v) :\; v\in N\}$,
\While{$A_i\ne \emptyset$}
\State Find a directed cycle $C_i$ in $D_i$. 
\If{$|C_i|=2$}
\State Let $e= E(C_i)$: $M:=M\cup e$. 
\State For each $v\in V(C_i):$ $k(v)=k(v)-1$
\Else
\If{$\exists v\in V(C_i): k(v)=1$}
\State For each $e\in E(C_i)$: $M:=M\cup \frac{1}{2}e$.
\State For each $v\in V(C_i):$ $k(v)=k(v)-1$
\Else
\State For each $e\in E(C_i)$: $M:=M\cup e$.
\State For each $v\in V(C_i):$ $k(v)=k(v)-2$
\EndIf
\EndIf
\State $V_{i+1}=\{ v\in V:\; k(v)\ge 1\}$
\State $A_{i+1}=\{ vp^M_{V_{i+1}}(v):\; v\in V_{i+1}\}$
\EndWhile
\end{algorithmic}
\end{algorithm}

\begin{Th}
\label{nearf}
Algorithm 1 produces a matching $M$ in $\mathcal{O}(|E|)$ time for the stable fixtures problem that is in the strong core of the instance with modified capacities, where $k(v)=\max \{ k(v),|M(v)|\} \le k(v)+ 1$.
\end{Th}
\begin{proof}
In each iteration we add at least one edge to $M$, so the algorithm terminates in at most $|E|$ iterations. 

Also, we only add at most two edges containing a given vertex $v$ in one step and only to vertices with $k(v)\ge 1$, so $|M(v)|\le k(v)+1$.

Finally we show that $M$ is in the strong core of this new instance. First of all it is easy to see, that if we run the algorithm with these new capacities we get the same output $M$, so we can suppose that the algorithm never violates the capacity constraints during its execution.

Assume that there is a blocking coalition $U$ for $M$ and let $M(U)$ be the matching for the vertices in $U$ that blocks $M$. Let $C_i$ be the first cycle that contains an edge that is not in $M(U)$, but contains a vertex of $U$ and let that vertex be $u$. Since $M(U)\not\subset M$, such a cycle exists. Then, by the fact that $u$ must have an at least as good partner set it $M(U)$ we get the edge corresponding to the arc $uw$ starting from $u$ in $C_i$ is in $M(U)$. ($u$ cannot get a better partner than $w=p^{M_i}_{V_i}(u) $ that she does not already had in $M$, because then there would be a cycle $C_j$ before $C_i$ that contains a vertex in $U$ but not every edge of $E(C_j)$ is in $M(U)$, a contradiction). This also means that $w\in U$, and by similar reasoning, the edge corresponding to the arc starting from $w$ is in $M(U)$, too, and continuing this argument we get that $E(C_i)\subset M(U)$, a contradiction.  
\end{proof}
\begin{Th}
For the stable fixtures problem Algorithm 2 produces a half-matching $M$ in $\mathcal{O}(|E|)$ time that is in the strong core of fractional matchings. 
\end{Th}
\begin{proof}
The running time is $\mathcal{O}(|E|)$, because in each iteration we add at least one edge fully to $M$ or two half edges.

The capacity constraints are obviously satisfied during the algorithm.

To show that the half-matching $M$ produced by the algorithm is in the strong fractional core, take the first cycle $C_i$ that has an agent from a blocking coalition $U$, but there is an edge of $C_i$ that has less weight in $M(U)$ than in $M$. Here $M(U)$ is the fractional matching that the agents of $U$ obtain among themselves to get a better solution.
Again, since $M$ restricted to $U$ cannot be $M(U)$, such a cycle exists. Let $u\in C_i\cap U$. The situation of $u$ has improved in the fractional matching $M(U)$, which means that the arc $uw$, $w=p_U^{M_i}(u)$ leaving $u$ must be in $M(U)$ with as much weight as possible (by the choice of $C_i$) and so are every edge of $C_i$. But this yields that all edges of $E(C_i)$ have the same weight in $M$ and $M(U)$, a contradiction.
\end{proof}

\subsection{Maximum size Pareto-optimal matchings}

Finally, we give an efficient algorithm for computing a maximum cardinality Pareto-optimal matching for the many-to-many stable matching problem. The techniques we use here are very similar to the ones in \cite{cechlarova2014pareto}, where they investigated Pareto-optimal matchings in the case when only one side of the agents have preferences.

Now we state our algorithm. We will call the two types of agents as men and women respectively. Denote the set of men as $U=\{ u_1,...,u_n\}$ and the set of women $W=\{ w_1,...,w_m\}$.

\begin{algorithm}\caption{Maximum size Pareto-optimal matching}
\begin{algorithmic}
\State Let $M:=\emptyset$
\For{$i=1,..,n$}
\State $l=1$
\While{$|M(u_i)|<k(u_i)$ and $l\le deg(u_i)$} 
\State Let $w_j$ be the $l$-th choice of $u_i$
\If{There is a maximum size matching containing $M\cup u_iw_j$}
\State $M:=M\cup u_iw_j$
\EndIf
\State $l=l+1$
\EndWhile
\EndFor
\end{algorithmic}
\end{algorithm}

\begin{Th}
For the many-to-many stable matching problem Algorithm 3 finds a maximum size matching that is Pareto-optimal in polynomial time.
\end{Th}
\begin{proof}
During the algorithm, each edge is checked at most once. Also, deciding if there exist a maximum size matching containing a given set of edges can be done in polynomial time (for example with the Hungarian method). So the running time of the algorithm is polynomial. 

Suppose that there is a matching $M'$, where each man obtains an at least as good partner set than in $M$ and there is a man $u_j$, who obtains strictly better. (If no such matching exists, then there cannot be any matching that Pareto-dominates $M$, since it would satisfy these conditions). 

Let the best partner of $u_j$ in $M'\setminus M$ be $w(u_j)$. If $w(u_j)$ is not saturated in $M$, then $u_j$ has to be saturated by the maximality of $M$. But then, there is an edge $u_jw_l\in M\setminus M'$, so letting $M''=M\cup u_jw(u_j)\setminus u_jw_l$, we obtain a maximum size matching, where each man has an at least as good situation and $u_j$ is strictly better off. However, this is a contradiction, since no maximum size matching can satisfy this, since among the maximum matchings, where $u_1,...,u_{j-1}$ obtain the same partner set, $M$ is one that is best for $u_j$ under lexicographical preferences. 

So for each man $u_j$, who is better off in $M'$, their best partner $w(u_j)$ is saturated in $M$. Hence from each such woman, there is an edge that is in $M\setminus M'$. Therefore, starting from $u_j$ and then, from each man going to their best partner in $M'\setminus M$ and from each woman, going to an arbitrary partner in $M\setminus M'$, we can find a cycle $C$. Each man in this cycle improves her situation in $M'$, so because of the lexicographicality of the preferences, they like the edge of $C\cap M'$ containing them better than the original ones. Thus letting $M''=M\cup (M'\cap C)\setminus (M\cap C)$, we obtain a maximum size matching, that dominates $M$ for each man, and at least one of them is strictly better off, a contradiction. 
\end{proof}

We remark that the above algorithm produces a matching that can be considered men-optimal in a sense. We could also run this algorithm with the women proposing instead of the men to obtain Pareto-optimal matchings that are "women-optimal", which yields a similar situation as in the original stable marriage problem.

\section*{Final notes}

We studied the strong core and Pareto-optimal solutions for multiple partners matching problems under lexicographic preferences. A natural question is to also consider the (weak) core. It is indeed an open question whether the core is always non-empty for this setting, and what is the computational complexity of finding such a solution. However, we believe that the concept of core can be too weak, that allows rather sub-optimal solutions. As a very simple example take three agents $a$, $b$ and $c$, where $a$ and $c$ has capacity one and $b$ has capacity two, and the possible pair are $ab$ and $bc$, where $b$ prefers $a$ to $c$. In this example the unique strong core, Pareto-optimal and stable solution is to take both edges. However, edge $ab$ alone is also in the core, since $bc$ is not blocking and the grandcoalition does not block either, since $a$ would not strictly improve.  

We worked with strict preferences in this paper, but some of our tractability results might be possible to extend for weak preferences. Finally, one could consider also to extend our efficient algorithms for finding certain solutions for more general additive and responsive preferences.

\bibliographystyle{plain}
\bibliography{lex}

\end{document}